\documentclass[aps,nofootinbib,amsfonts,superscriptaddress]{revtex4}

\usepackage{amsfonts,amssymb,amscd,amsmath}
\usepackage{graphicx}
\usepackage{subfigure}
\usepackage{color}
\DeclareGraphicsExtensions{.pdf,.jpg,.png,.gif}
\usepackage{tikz-cd}
\usepackage{tikz}
\graphicspath{{figures/}}
\usepackage[utf8]{inputenc}
\usepackage{amsthm}

\newtheorem{theorem}{Theorem}

\newcommand{\beql}[1]{
  \begin{equation}\label{eq:#1}}
\newcommand{\eeq}{
  \end{equation}}

\definecolor{brightpink}{rgb}{1.0, 0.0, 0.5} 

\begin{document}

\author{Niels Gresnigt}
\email{niels.gresnigt@xjtlu.edu.cn}
\affiliation{Department of Physics, Xi'an Jiaotong-Liverpool University, 215123 Suzhou, China}

\author{Antonino Marcian\`o}
\email{marciano@fudan.edu.cn}
\affiliation{Center for Field Theory and Particle Physics \& Department of Physics, Fudan University, 200433 Shanghai, China}
\affiliation{Laboratori Nazionali di Frascati INFN, Frascati (Rome), Italy, EU}

\author{Emanuele Zappala}
\email{zae@usf.edu, emanuele.amedeo.zappala@ut.ee}
\affiliation{Institute of Mathematics and Statistics, University of Tartu, Tartu, Estonia, EU}

\title{Braided matter interactions in quantum gravity via $1$-handle attachment}

\begin{abstract}
\noindent
In a topological description of elementary matter proposed by Bilson-Thompson, the leptons and quarks of a single generation, together with the electroweak gauge bosons, are represented as elements of the framed braid group of three ribbons. By identifying these braids with emergent topological excitations of ribbon networks, it has been possible to encode this braid model into the framework of quantum geometry provided by loop quantum gravity. One major hurdle in this promising approach to unifying matter and spacetime has been the difficulty in implementing any dynamics that reflects observed particle interactions. In the case of trivalent networks, it has not been possible to generate particle interactions, because the braids correspond to noiseless subsystems, meaning they commute with the evolution algebra generated by the local Pachner moves. In the case of tetravalent networks, interactions are only possible when the model's original simplicity, in which interactions take place via the composition of braids, is sacrificed. The main result of the present paper is to demonstrate that it possible to preserve both the original classification of fermions, as well as their interaction via the braid product, if we embed the braid in a trivalent scheme, and supplement the local Pachner moves, with a non-local and graph changing 1-handle attachment. Moreover, we use Kauffman-Lins recoupling theory to obtain invariants of braided networks that distinguish topological configurations associated to particles in the Bilson-Thompson model.

\end{abstract}

\maketitle

\section{Introduction}
\label{Sec:Introduction}

In the topological matter model proposed by Bilson-Thompson \cite{BT}, the leptons and quarks of the Standard Model (SM) are interpreted as simple braids of three twisted ribbons, bound together at the top and bottom by a parallel disk. The three ribbons are allowed to braid each other, which allows the ribbons to be distinguished by their relative crossings. With the ribbons distinguished in this way, the \textit{twist structure} of the ribbons accounts for the electrocolor symmetries. The \textit{braid structure} on the other hand encodes the weak symmetry and chirality, with a top to bottom reflection corresponding to mapping between particles and anti-particles, and a left to right reflection corresponding to a parity transformation. The electroweak interactions are proposed to be represented topologically via braid composition (such a process however ignores the disks to which the three ribbons are connected at both the top and bottom). \\

Although the model provides an appealing way to encode the quantum numbers associated with leptons and quarks in terms of simple topological features, and represents their electroweak interactions via the simple topological process of braid composition, one missing element of the model is an explanation of where these braids live. Furthermore, the model does not provide a dynamical framework. \\

In loop quantum gravity (LQG), the states correspond to spin networks. In the case of a non-zero cosmological constant, the edges of these spin networks are framed, generalizing the spin network to a ribbon network, in which (for trivalent networks) the edges and vertices become ribbons and disks respectively. Such ribbon networks provide a natural home for the braids in \cite{BT}, and embedding these braids into ribbon-networks could address one of the weaknesses of LQG as a theory for unification; namely the lack of a particle spectrum. \\
Indeed, soon after the proposal in \cite{BT}, it was shown that the proposed braided states correspond precisely to the simplest emergent conserved topological excitations of ribbon networks in a background topological space (three manifold), thereby providing an embedding of Bilson-Thompson's model into quantum gravity theories in which states are labeled by diffeomorphism classes of embeddings of ribbon graphs in a background topological space \cite{BT2}. \\ 

The dynamics in these theories are generated by local evolution moves on the graphs. These evolution moves on trivalent spin networks are dual to the Pachner moves that relate different triangulations of the same surface. In the case of ribbon-networks, these moves are suitably adopted to account for the framing of the spin networks. \\

In this trivalent setting, a braid occurs as a single trivalent node from which three ribbons emerge. The ribbons twist and braid around each other, and then join to the rest of the network at three other trivalent nodes\footnote{Braids, as sub-graphs of extended 1-complexes, require at least four trivalent nodes to be connected to the embedding graph. Differently, isolated tangles can be achieved by interconnecting two trivalent nodes.}. One then hopes that the dynamics governing particle interactions arises as a consequence of the dynamics of ribbon networks. However, one finds that these subsystems of ribbon networks correspond to noiseless subsystems, familiar from quantum information theory. Noiseless subsystems correspond to subsets of states of a quantum system that are preserved under the evolution algebra of the system. In the present context this means that under the standard local exchange and expansion evolution moves, the braiding and twisting content of these subsystems are preserved (they commute with the evolution moves), and therefore provide conserved degrees of freedom and associated quantum numbers.\\

As noiseless subsystems, the elementary particles are protected by symmetries in the dynamics, and therefore remain coherent as they propagate. Although this is consistent with the observed stability of these particles, at the same time it prevents the emergence of genuine particle interactions, therefore hampering the development of a fully dynamical theory, because the states are too strongly conserved. In order to introduce a nontrivial dynamics for braided particles, it is necessary to consider moves such that the braided particles do not constitute noiseless subsystems.\\

%

To overcome the serious limitation of braids being conserved too strongly, the tetravalent scheme was developed. The topological excitations to be identified with particles are still three ribbon braids, but now formed by the three common edges of two adjacent tetravalent nodes. The braids can therefore be considered as an insertion in an edge of the ribbon graph. The subsequent dynamics based on the (adapted) dual Pachner moves gives rise to forms of braid propagation and interaction that are in certain instances analogues to the dynamics of particles. This is because in the tetravelent scheme, the embedded braids no longer correspond to noiseless subsystems, meaning that a braid's structure can undergo changes during its propagation. In an interaction of two adjacent braids, one can merge with the other through a series of evolution moves. This braid merger however does not correspond to the familiar composition of braids. Furthermore, in the tetravalent scheme one loses an obvious identification between specific braids and the leptons and quarks of the SM. In the trivalent scheme \cite{BT} maps the trivalent braids to SM fermions. However the tetravalent case generates an infinite range of equivalence classes of braids, and sufficient super-selection rules to select appropriate braids to map to SM fermions are lacking. \\

Both the trivalent and tetravalent scheme have their advantages, whilst at the same time introducing limitations. Although the trivalent scheme successfully establishes a correspondence between braids and SM particles, and the tetravalent scheme provides a dynamical theory of interactions ruled by topological conservation laws, neither scheme does both. The trivalent scheme lacks dynamics, where the tetravalent scheme is plagued by an infinite range of equivalence classes of braids. In particular, neither theory is capable of representing the electroweak interaction in the spirit of \cite{BT}, that is via the braid product. A comprehensive review of both the trivalent and tetravalent scheme can be found in \cite{BT3}.\\

In this paper, we consider the trivalent scheme and generate interactions that correspond topologically to braid compositions by complementing the Pachner moves by a new move corresponding to adding a 1-handle. This new move, and hence the particle interaction, is not a local one (like the dual Pachner moves) and so two particles, separated in the ribbon network are able to interact at the topological level. Attaching handles corresponds to having a $4$-dimensional cobordism from an initial state embedded into a $3$-dimensional manifold to a final state in a $3$-dimensional manifold. This cobordism is obtained from crossing the base space, which is $3$-dimensional, with a time coordinate. The construction is therefore inherently $4$-dimensional. Unlike the Pachner moves, which are graph changing but leave the underling manifold unchanged, attaching a 1-handle changes the underlying manifold. Our proposal restores one of the most attractive features of Bilson-Thompson's model, that of representing particle interactions topologically as braid composition, whilst at the same time allowing the model to be embedded into background independent theories of quantum gravity. The identifications of leptons and quarks correspond to those in \cite{BT}, and electroweak interactions are represented topologically via the braid product. This is the first main result of this paper. \\

Finally we argue that, instead of using a link invariant to distinguish between different particles, as has been done previously \cite{BT2}, the algebraic invariant of isotopy classes of braided networks based on Kauffman-Lins recoupling theory \cite{KL} provides a more appropriate choice invariant, being able to distinguish between the different particles. This is the second main result of this paper. \\

\section{A topological model of composite preons}
\label{sec:Braidmodel}
We begin by providing a brief overview of the Bilson-Thompson model \cite{BT}. The fundamental preonic object is a ribbon which may be twisted by $\pm 2\pi$. Ribbons are combined into triplets with non-trivial braiding, and connected together at the top and bottom to a parallel disk. Such a braid constitutes a two-dimensional surface. It is assumed that ribbons with twists in opposite directions are not permitted in the same triplet. The resulting topological objects, corresponding to framed braids in the circular Artin braid group $B_3^c$, are called 3-belts. With a simple (but arbitrary) choice for how the ribbons are braided, the braids in Figure \ref{helonmodel}, represent the first generation of SM leptons and quarks.
\begin{figure}[h!]
\centering
 \includegraphics[scale=0.5]{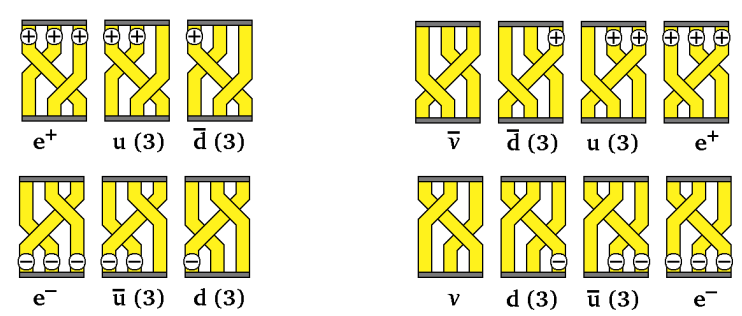}
\caption{\footnotesize{In the braid model, leptons and quarks are represented as braids of three (possibly twisted) ribbons. Charged fermions come in two handedness states whereas the neutrino and antineutrino come in only one handedness state. Source \cite{BT}} \label{helonmodel}}
\end{figure}
In this representation, the \textit{twist structure} of the ribbons accounts for the electrocolor symmetries, with charges of $\pm e/3$ represented by $\pm 2\pi$ twists, and the permutations of twisted ribbons representing color. This braid model thus allows one to describe electric charge and colour charge in terms of the topological structure of braids. Without any underlying braiding of the three ribbons, a simple rotation of both end disks means the ribbons become indistinguishable from one another, and one loses the ability to represent the color symmetry. No explicit assumption is made about the permitted braiding of ribbons, and a well-motivated choice for the braid structure to assign to leptons and quarks is still lacking\footnote{Some recent works by one of the present authors demonstrates that the topological structures of the helon model can be generated from the minimal ideals of the two complex Clifford algebras $\mathbb{C}\ell(6)$ and $\mathbb{C}\ell(4)$, by establishing a map between the basis states of these ideals, and products of braid generators in the circular Artin braid group $B_3^c$ (for $\mathbb{C}\ell(6)$) and $B_3$ (for $\mathbb{C}\ell(4)$). In such a construction, the finite-dimensionality of the minimal ideals provides a natural limitation for the complexity of twisting of ribbons, as well as the braiding of ribbons \cite{NGG1,NGG2,NGG3}.}. The model in Figure \ref{helonmodel} chooses just one simple possibility. More complex braiding may be speculated to correspond to additional generations. However, without any restriction on the complexity of braiding, this leads to an unbounded number of generations.\\

The model does not provide a dynamical framework. Weak interactions however may be represented topologically via braid composition, as shown in Figure \ref{composition}, where the electroweak bosons are assumed to correspond to unbraids (with possible twisting on its constituent ribbons). Forming the braid product requires that the bottom disk of the first composing braid, and the top disk of the second composing braid be removed. One limitation of the original model is that this issue is not addressed. In section \ref{Surgery} we show that attaching a 1-handle naturally resolves this issue.
\begin{figure}[h!]
\centering
\includegraphics[scale=0.5]{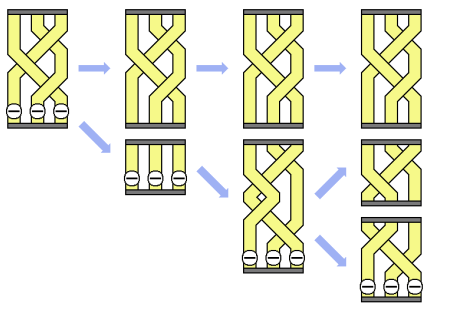}
\caption{\footnotesize{Weak interactions are represented topologically via braid composition. Source \cite{BT}}} \label{composition}
\end{figure}

\begin{figure}[h!]
\centering
 \includegraphics[scale=0.4]{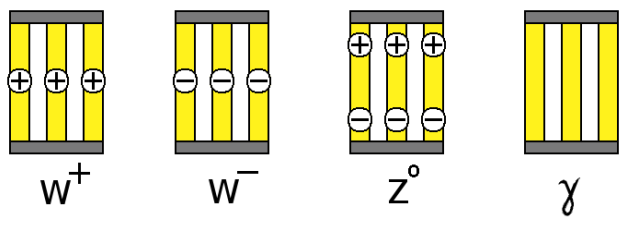}
\caption{\footnotesize{Electroweak bosons are assumed to correspond to triplets of ribbons that are not braided. Source \cite{BT}}} \label{bosons}
\end{figure}




\section{Theories on ribbon graphs}

We now define and summarise the class of theories that we focus on. Much of what follows in this section can be found discussed at greater length in \cite{BT2}. We are interested in theories of ribbon graphs. These relate closely to LQG and spin foam models, but in which the graphs that comprise the states in these theories are framed. The states of these theories then are represented by the embedding of a two surface into the spatial manifold. \\

A two-dimensional surface (with boundary) $S$ of genus at least two in a compact three manifold $\Sigma$, can be considered as a union of trinions, which are surfaces with three distinct region of connection to other surfaces. These trinions ultimately represent the framing of trivalent spin network nodes. A ribbon graph $\Gamma$ then corresponds to a particular decomposition of $S$. \\

By associating a finite-dimensional state space $\mathcal{H}_{t}$ to each trinion, and the tensor product operation to the glueing of trinions, it is possible to construct a map from the ribbon graph to a quantum system. In this paper we do not focus on the coloring of the graph, and therefore ignore the association of a (quantum) group element to each edge of a trinion. Instead it suffices to use a functor to assign a finite-dimensional vector space to each trinion. The state space associated with the entire graph $\Gamma$ is then given by
\begin{equation}
    \mathcal{H}_{\Gamma}=\bigotimes_{t\in\Gamma} \mathcal{H}_t,
\end{equation}
where the label $t$ runs over all trinions. The state space of the theory is subsequently a sum over all the topologically distinct embeddings of all such ribbon graphs in $\Sigma$ with inner product $\langle \Gamma | \Gamma' \rangle=\delta_{\Gamma\Gamma'}$
\begin{equation}
    \mathcal{H}=\bigoplus_{\Gamma}\mathcal{H}_{\Gamma}.
\end{equation}

\subsection{The evolution algebra}

By embedding the braids representing the leptons and quarks in \cite{BT} in this manner, it becomes possible to define local dynamics by excising subgraphs of the ribbon network $\Gamma$ and replacing them with new ones. Because trivalent spin networks are the dual skeletons of triangulations of 2D surfaces in which a node is dual to a 2-simplex, the evolution moves of trivalent spin networks which generate such dynamics are dual to the Pachner moves that relate triangulations. In the case of a ribbon network, these Pachner moves need to be suitably adapted. There are three non-trivial generators $A_i$ of evolution, which together with the identity generate the evolution algebra on the sate space $\mathcal{H}$
\begin{equation}
    \mathcal{A}_{evol}=\{1, A_i\},\quad i=1,2,3\,,
\end{equation}
For a given ribbon graph $\Gamma$, the application of $A_i$ results in
\begin{equation}
    \hat{A_i}|\Gamma \rangle=\sum_\alpha |\Gamma'_{\alpha i}\rangle,
\end{equation}
where $\Gamma'_{\alpha i}$ represent the ribbon graphs that can be obtained from $\Gamma$ via an application of one of the evolution moves. \\

A key observation in \cite{BT2} is that any braiding and twisting content of states in the state space $\mathcal{H}$ remains invariant under such evolution of the ribbon graph, as generated by the evolution algebra. This means that any physical information that is encoded in the braiding will necessarily be conserved under evolution. In other words, the evolution algebra only acts non-trivially on the trinions of the graph, whereas any braids embedded into the ribbon graph are noiseless, commuting with the evolution algebra. 

\section{Leptons and quarks as topological excitations}

\subsection{The trivalent scheme}

What was shown in \cite{BT2} is that the braids proposed in \cite{BT}, representing leptons and quarks, correspond to the simplest conserved topological excitations in a trivalent network. By conserved, we mean that the generators of such dynamics (the adapted dual Pachner moves) commute with the braiding and twisting content of the topological excitations, and the latter therefore correspond to noiseless subsystems of the quantum dynamics. This means that although the simplest excitations propagate coherently, they are too strongly conserved since they do not destabilise \cite{BT2}. Thus, although it is possible to describe scattering, the creation and annihilation of particles is impossible  without some  modification \cite{JH}.

\subsection{The tetravalent scheme}

To overcome this serious limitation, and because of the geometrical correspondence between framed tetravalent spin networks and 3-space, a tetravalent scheme was subsequently developed \cite{YW,SW}. In this scheme, the topological excitations to be identified with particles are still three strand braids, but now formed by the three common edges of two adjacent tetravalent nodes. The braids can therefore be considered as insertion in an edge of the ribbon graph. 
\\

In this case, the dynamics based on the (adapted) dual Pachner moves naturally associated with tetravalent graphs gives rise to forms of braid propagation and interaction that are in some instances analogous to the dynamics of particles. More precisely, the theory gives rise to different classes of braids. One class of braids neither propogates nor interacts. A second class of braids propagate but do not interact. In this case, one finds that the propagation is chiral, meaning that some braids can only propagate to their right in relation to the local subgraph, whilst others only propagate to their left. Such braids are speculated to correspond to fermions. A final class of braids consists of actively interacting braids that are two-way propagating. These braids are capable of merging with neighboring braids when certain interaction conditions are met, and are speculated to correspond to bosons \cite{SW}. Exchanges of such actively interacting braid excitations give rise to the interactions between fermionic braids.\\

Dynamics generated in this manner strongly constraints the number of possible discrete transformations in the tetravalent scheme to be exactly seven (excluding the identity), corresponding to C, P, T, and their product, and the interactions of braids are found to be invariant under C, P, and T separately, and thus also under CPT \cite{HW2}. Braid Feynman diagrams have been developed and used to represent the dynamics of braids, with the hope that an effective theory describing braid dynamics can be based on these braid Feynman diagrams \cite{YW2}.\\

As the above paragraphs indicate, the tetravalent formalism has proven to be more adapt at generating dynamics than the trivalent formalist, but not without introducing a new set of challenges to overcome. In the tetravelent scheme one lacks a means to pick out those tetravalent braids that should be mapped to the SM particles. Whereas in the trivalent case, such a map is provided by the model in \cite{BT}, with each equivalence class of braids being mapped to a single type of particle, the tetravalent case produces a seemingly infinite range of (equivalence classes of) braid states analogous to bosons and fermions.\\

\section{Ribbon networks embedded into 3-manifolds}

We suppose that particle systems constitute {\it ribbon graphs} in which vertices are thickened to be (homeomorphic to) two dimensional disks (or rectangles), and edges are thickened to ribbons. These objects are embedded in a 3-dimensional ground space, which is taken to be a 3-manifold, and they are also referred to as {\it fat graphs} in the literature. See for instance \cite{Igusa,MP}.   \\

Ribbon graphs are formally defined as oriented surfaces decomposed into two families of rectangular surfaces, i.e. embeddings of the square $[0,1]\times [0,1]$ in the base 3-dimensional manifold, called ``{\it coupons}'' and ``{\it ribbons}'', and a family of ``{\it annuli}'', i.e. embeddings of the cylinder $S^1\times [0,1]$. We will think of coupons as vertices of a graph, thickened into rectangles or, equivalently a $2$-dimensional closed ball, while ribbons will constitute the connecting edges of a graph, thickened to rectangles. Annuli determine the components of framed links, and will therefore be given with a trivialization of the normal bundle, which is described by an integer corresponding to the twisting. Moreover, it is assumed that coupons, ribbons and annuli do not self-intersect, but coupons and ribbons are allowed to meet at one side of the boundary of the rectangular surface. Figure~\ref{fig:ribbongraph} shows a portion of ribbon graph where a central coupon, rectangular and corresponding to a ``fat vertex'' of the graph, is connected to ribbon edges on top and bottom. This is the fundamental unit of a ribbon network considered below, as we can concatenate them by joining their ribbons. \\

\begin{figure}
    \centering
    \begin{tikzpicture}
		
		\draw (0,2) -- (1,0);
		\draw (0.2,2) -- (1.2,0);
		\draw (0,2)--(0.2,2);
		
		\draw (1.5,2) -- (1.5,0);
		\draw (1.7,2) -- (1.7,0);
		\draw (1.5,2)--(1.7,2);
		
		\draw (3,2) -- (2,0);
		\draw (3.2,2) -- (2.2,0);
		\draw (3,2)--(3.2,2);
		\draw[dashed] (0.5, 1.8) -- (1.4,1.8); 
		\draw[dashed] (1.8,1.8) -- (2.7,1.8);
		\draw (0,0) -- (3.2,0);
		\draw (0,-1) -- (3.2,-1);
		\draw (0,0) -- (0,-1);
		\draw (3.2,0) -- (3.2,-1);
		\draw (1,-1) -- (0,-3);
		\draw (1.2,-1) -- (0.2,-3);
        \draw (0,-3)--(0.2,-3);
		
		\draw (1.5,-1) -- (1.5,-3);
		\draw (1.7,-1) -- (1.7,-3);
		\draw (1.5,-3)--(1.7,-3);
		
		\draw (2,-1) -- (3, -3);
		\draw (2.2,-1) -- (3.2,-3);
		\draw (3,-3)--(3.2,-3);
		\draw[dashed] (0.5, -2.8) -- (1.4, -2.8);
		\draw [dashed] (1.8,-2.8) -- (2.7,-2.8);
		
			\end{tikzpicture}
    \caption{Example of ribbon graph. The top ribbons and bottom ribbons have one side connected to the central rectangular coupon, homeomorphic to a disk}
    \label{fig:ribbongraph}
\end{figure}
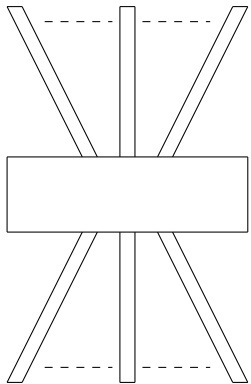

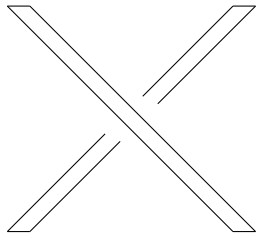
\begin{figure}
    \centering
    \begin{tikzpicture}
		\draw (0,3) -- (3,0);
		\draw (0.3,3) -- (3.3,0);
		
		\draw (3,3) -- (1.8,1.8);
		\draw (1.3,1.3) -- (0,0);
		
		\draw (3.3,3) -- (2,1.7);
		\draw (1.5,1.2) -- (0.3,0);
		
		\draw (0,3) -- (0.3,3);
		\draw (3,3) -- (3.3,3);
		\draw (0,0) -- (0.3,0);
		\draw (3,0) -- (3.3,0);
		
				\end{tikzpicture}
    \caption{A positive crossing of two ribbons, oriented downward}
    \label{fig:crossing}
\end{figure}
Two ribbon graphs are said to be equivalent if there exists an isotopy that transforms one into the other. Equivalence of ribbon graphs is conveniently expressed in combinatorial terms by considering diagrams of ribbon graphs, i.e. projections on the plane. In terms of their diagrams, two ribbon graphs are equivalent if and only if there exists a finite sequence of plane isotopies and moves of type:
\begin{itemize}
    \item 
    Oriented framed Reidemeister moves;
    \item 
    Ribbons can slide above and below coupons;
    \item
    Opposite crossings cancel, after orientation turning is applied;
    \item
    Coupons can be rotated with surface parallel to the plane, so that the attached ribbons spiral around the coupons.
\end{itemize}
This combinatorial moves correspond to the generators of the category of ribbon graphs in \cite{RT}, relations ${\rm rel}_1$ to ${\rm rel}_{13}$. \\

In analogy with spin networks in quantum gravity, we call {\it ribbon networks} the elements of an appropriate subclass of ribbon graphs (introduced below) considered in this article. We assume that our ribbon networks are embedded into $3$-manifolds. A $3$-manifold $M$ is thought of as a temporal fragment of space-time, in the sense that it represents the spatial coordinates at a fixed time. A network embedded in $M$ represents a physical system at time $t_0$. Figure~\ref{fig:embeddednetwork} illustrates a portion of a trivalent network embedded into a $3$-manifold corresponding to an arbitrary braid $B$ of Bilson-Thompson's model.\\ 

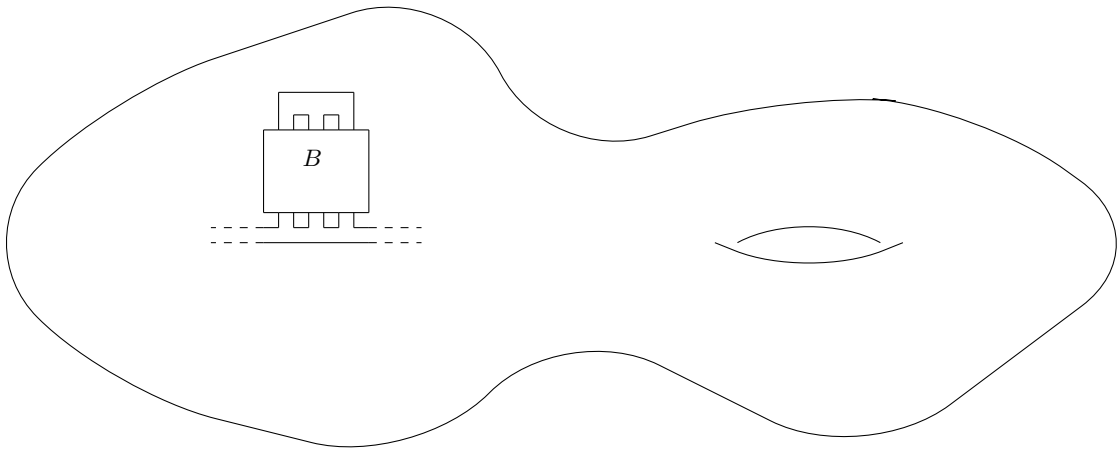
\begin{figure}
    \centering
    \begin{tikzpicture}
					\draw[rounded corners=40pt](6,-2)--(4,-1)--(2,-3)--(-2,-2)--(-4,0)--(-2,2)--(2.5,3.5)--(3.8,1)--(7,2)--(9.5,1.8)--(12,0)
					--(8,-3)--(6,-2);
					
					\draw[rounded corners=30] (6,0)--(7.25,-0.5)--(8.5,0); 
					\draw[rounded corners=30] (6.3,0)--(7.25,0.5)--(8.2,0);
					
					\draw (0,0)--(1.4,0);
					\draw (0,0.2)--(0.2,0.2);
					\draw (0.4,0.2)--(0.6,0.2);
					\draw (0.8,0.2)--(1,0.2);
					\draw (1.2,0.2)--(1.4,0.2);
					
					\draw (0.2,0.2)--(0.2,0.4);
					\draw (0.4,0.2)--(0.4,0.4);
					\draw (0.6,0.2)--(0.6,0.4);
					\draw (0.8,0.2)--(0.8,0.4);
					\draw (1,0.2)--(1,0.4);
					\draw (1.2,0.2)-- (1.2,0.4);
					
					\draw (0.2,2)--(1.2,2);
					\draw (0.2,1.5)--(0.2,2);
					\draw (1.2,1.5)--(1.2,2);
					
					\draw (0.2,1.5)--(0,1.5);
					\draw (1,1.5)--(1.4,1.5);
					\draw (0,0.4)--(0.2,0.4);
					\draw (1,0.4)--(1.4,0.4);
					\draw (0,0.4)--(0,1.5);
					\draw (1.4,0.4)--(1.4,1.5);
					\draw (0.2,0.4)--(1,0.4);
					\draw (0.2,1.5)--(1,1.5);
					
					\draw (0.4,1.5)--(0.4,1.7);
					\draw (0.6,1.5)--(0.6,1.7);
					\draw (0.8,1.5)--(0.8,1.7);
					\draw (1,1.5)--(1,1.7);
					\draw (0.4,1.7)--(0.6,1.7);
				    \draw (0.8,1.7)--(1,1.7);
					
					\draw[dashed] (0,0)--(-0.7,0);
					\draw[dashed] (0,0.2)--(-0.7,0.2);
					
					\draw[dashed] (1.4,0)--(2.1,0);
					\draw[dashed] (1.4,0.2)--(2.1,0.2);
					
					 \node (a)  at (60:1.3) {$B$};
					\end{tikzpicture}
    \caption{Fragment of trivalent ribbon network embedded in a $3$-manifold. The box named ``$B$'' corresponds to a braid of Bilson-Thompson model}
    \label{fig:embeddednetwork}
\end{figure}

 We take our ``admissible'' ribbon networks to be closed trivalent ribbon graphs, i.e. ribbon networks with no open edges whose vertices have three incident edges. We allow, moreover, disconnected annuli, i.e. loops, knotted with the rest of a ribbon network. In other words, an admissible ribbon network is a compact (not necessarily connected) surface with boundary whose corresponding spine is a trivalent graph embedded in $3$-space. A vertex coupon is represented as in Figure~\ref{fig:trinion}, and will be called a {\it trinion}. We consider isotopy classes of braided networks, and therefore do not distinguish between representatives of the same class. An isotopy class is determined combinatorially by a slightly modified version of the moves described above. In fact, the set of moves relating diagrams of isotopic surfaces with boundary embedded in a $3$-manifold has been determined in \cite{Matsuzaki} in both the oriented and non-oriented case. Since Bilson-Thompson's framed braids are all oriented, it is reasonable to require that our ribbon graphs, and hence ribbon networks, are all taken to be oriented. In the oriented case, Matsuzaki's moves \cite{Matsuzaki} correspond to the framed Reidemeister moves, as given for example in \cite{Ohtsuki}, Figure~1.8, slide of trinions over and underneath ribbons, and Ishii's IH move in \cite{Ishii}. In fact, the latter move has been considered in \cite{BT2} under the name of exchange move. We remark that using Matsuzaki's moves, in principle, the procedure of trading braiding for twists discussed in \cite{BT2009} is not allowed, as the oriented moves do not include vertex twists. Roughly speaking, we always look at one side of the ribbons, without allowing half-twists. This is consistent with the constructions found in this article, as a consequence of restricting ourselves to oriented surfaces.\footnote{As appeared in Bilson-Thompson's work, full twists are compositions of two half-twists, either positive or negative. In the present context, what is meant by a full twist is a loop with (positive or negative) self-crossing (also called {\it kink}), as in Matsuzaki's moves.} We will see below that it is useful to derive a relation between self-crossing twists and composition of half-twists in order to simplify computations. We will use this relation as a computational tool, although no isolated half-twist appears in our framework. \\

\begin{figure}
    \centering
    \begin{tikzpicture}
 	\draw[rounded corners=20pt](-1,1)--(0,-1)--(0,-2);
 	\draw[rounded corners=20pt](1,1)--(0.5,0)--(0,1); 
 	\draw[rounded corners=20pt](1,-2)--(1,-1)--(2,1);
 	\draw (-1,1)--(0,1);
 	\draw (1,1)--(2,1);
 	\draw (0,-2)--(1,-2);
 	\end{tikzpicture}
    \caption{A trivalent vertex is represented by a trinion}
    \label{fig:trinion}
\end{figure}
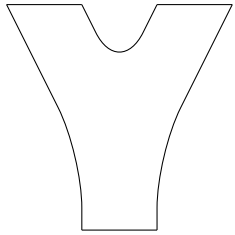
	
Lastly, recall that each $3$-manifold arises from surgery on a framed link embedded in the $3$-sphere $\mathbb S^3$, see for instance \cite{Lick}. This result, known in the literature as Lickorish-Wallace Theorem, allows us to represent a $3$-manifold $\mathcal M^3$ by a framed link $\mathcal L$ whose $\mathbb S^3$-surgery produces $\mathcal M_{\mathcal L}^3$. Using the definition of braided network, a framed link $\mathcal L$ is indeed a (possibly part of) braided network. We can therefore represent both base $3$-space and braided matter by means of braided networks as the disjoint union of a framed link $\mathcal L$ representing $\mathcal M_{\mathcal L}^3$, the $3$-dimensional base space, and matter embedded in $\mathcal M^3$ is encoded in an embedded surface with non-trivial number of trinions. In other words, space-time and matter are considered on the same footing.\\

The set of diagrammatic moves relating two framed links that correspond to diffeomorphic $3$-manifolds were obtained by Kirby \cite{Kir}, and are referred to as {\it Kirby's Calculus}. A braided network $\mathcal N$, therefore, is presented as union $\mathcal N = \Gamma\cup \mathcal L$ of a trivalent surface with boundary $\Gamma$ and a framed link $\mathcal L$, where $\mathcal L$ represents (via Lickorish-Wallace Theorem and link surgery) the base $3$-space $\mathcal M^3_{\mathcal L}$, and $\Gamma$ gives the braided matter embedded in $\mathcal M^3_{\mathcal L}$.

\subsection{Surgery on $3$-manifolds}\label{Surgery}

A cobordism is an ordered triple of manifolds \cite{Kos}, written $\{X_0,Y,X_1\}$, such that $\partial Y = X_0\bigsqcup X_1$. If the manifolds are taken to be oriented, one assumes further that the orientations of $X_i$ and the orientations induced on them by that of $Y$ coincide. \\

Recall first, the definition of handle attachment, which we hereby provide. Let $X$ be a manifold of dimension $n$, and let $\psi$ be an embedding of $\partial B^p\times B^{n-p}$ in $\partial X$, where $B^k$ indicates the standard $k$-dimensional closed ball. We say that we attach a $p$-handle to $X$ when we attach a copy of $B^p\times B^{n-p}$ along the image $\psi(\partial B^p\times B^{n-p})$. The notion of handle attachment is closely related to that of surgery on a manifold. Suppose that we have an embedding of the $k$-dimensional sphere $\phi: \mathbb S^k \longrightarrow M$, into the $n$-dimensional manifold $M$. A framing of $\phi(\mathbb S^k)$ is a (continuous) section of the normal bundle of $\phi(\mathbb S^k)$. An embedded sphere, along with a framing of its image, determine a map $\psi: \mathbb S^k\times \mathbb B^{n-k}\longrightarrow M$ up to isotopy. A $k$-surgry on $\psi$ consists of removing $\psi(\mathbb S^k\times \mathbb B^{n-k})$ from $M$, and replacing it by a copy of $\mathbb B^{k+1}\times \mathbb S^{n-k-1}$ via the map $\psi$. It is known (\cite{Kos,4Manifolds}) that a $k$-handle attachment on a $(n+1)$-dimensional manifold with boundary is equivalent to a $(k-1)$-surgery on the $n$-dimensional boundary. We will therefore, sometimes, improperly use the two names interchangeably.  \\

Let now $X$ be a manifold and consider the cylinder $X\times I$, where $I:= [0,1]$ is the closed unit interval, and let us indicate by $X_i$ the copy $X\times \{i\}$ of $X$ in $X\times I$, with $i=0,1$. We attach a $p$-handle to the boundary of $X\times I$ along a copy of $\partial B^p\times B^{n-p}\hookrightarrow X_1$. It is a known results that the manifold so obtained, is a cobordism $Y$ between $X$ identified with the copy $X_0$ of $X\times I$, and $X_1$ with surgery along the attaching sphere of the $p$-handle. We say that such a cobordism arising from attaching a handle to a given manifold is an {\it elementary cobordism}. \\

Recall, further, that given two cobordisms $\{X_0,Y,X_1\}$ and $\{X'_0,Y',X'_1\}$ where we have a diffeomorphism $X_1 \cong X'_0$ (and compatible orientations in the orientable setting), we can construct a ``composition'' cobordim $\{X_0,Y\bigcup_{X_1}Y',X'_1\}$, obtained by gluing $Y$ and $Y'$ along their (diffeomorphic) common boundary for a given choice of diffeomorphism. This operation of composition of cobordisms can clearly be repeated. We have the following traditional result, proof of which can be found in \cite{Kos}, Chapter $7$. 

\begin{theorem}\label{thm:cobordism}
Let $\mathcal C = \{X_0,Y,X_1\}$ be a cobordism. Then $\mathcal C$ can be obtained as:
$$
\mathcal C = \mathcal C_1 \cup \cdots \cup\mathcal C_t,
$$
for elementary cobordisms $\mathcal C_i = \{X^i_0,Y^i,X_1^i\}$, $i=1,\ldots , t$. 
\end{theorem}

In fact, in our setting we are not interested in all types of elementary cobordisms and, consequently, general types of cobordisms. We restrict ourselves to handles attachment of degree 1 to $4$-manifolds with $3$-dimensional boundary, according to certain specific rules which we now describe. Firstly, observe that if $M$ is a $3$-dimensional manifold, then $\partial B^p\times B^{4-p} = \partial B^1\times B^3$, which is diffeomorphic to two $3$-dimensional closed balls. Then, the $1$-handle attached to $M\times I$ following the procedure described above, is a copy of $B^1\times B^3$, i.e. a cylinder whose section is a $3$-dimensional ball. 

\subsection{Braid dynamics via surgery}\label{Dynamics}

Let us consider now a trivalent network $\mathcal N$ embedded in a $3$-manifold $M$. An elementary cobordism obtained via $1$-handle attachment is said to be {\it admissible} if the two embedded $3$-balls $\partial B^1\times B^3$ intersect two braids $b_1, b_2$, representing particles as in the Bilson-Thompson model, in such a way that $\partial B^1\times \partial B^2$ intersect $b_1$ and $b_2$ in exactly two triplets of disks and, moreover, the portions of $b_1$ and $b_2$ contained inside ${\rm int}(\partial B^1\times B^3)$ are isomorphic to the ball with a trivial trivalent braid attached to it, i.e. an embedded trinion. Once we attach a $1$-handle, the new manifold  $M'$ so obtained contains two open braids, which we connect via straight ribbons contained in the attached handle, in such a way that no new crossing or twisting is introduced. We therefore obtain a new manifold $M'$ containing an embedded network $\mathcal N'$ that is identical to $\mathcal N$ outside the attached $1$-handle. By construction, $\mathcal N'$ is obtained by joining $b_1$ and $b_2$ following Bilson-Thompson composition rules described in Section~\ref{sec:Braidmodel} above. The manifolds $M$ and $M'$ containing $\mathcal N$ and $\mathcal N'$, respectively, are related by an elementary (admissible) cobordism. Of course, we allow also the opposite procedure, in which we cut a handle containing a braid, with the condition that no crossing or twisting is contained in the handle, and attach back two $3$-balls in which we embed two disks along with (unbraided) ribbons joined to the braids $b_1$ and $b_2$ obtained by cutting the handle. The two processes described above, take the physical interpretation of particle fusion and decay. \\

We define the dynamical interaction in a braided network to be restricted to admissible cobordisms and merging of braids induced by them, according to the scheme described above. As a consequence, a system represented by an embedding of a network $\mathcal N$ in $M$, can evolve only into a network $\mathcal N'$ which differs from $\mathcal N$ only for a finite number of merging/cutting of braids. Observe that the dynamics inherently presents a $4$-dimensional component due to the fact that the cobordisms induced by surgery are $4$-dimensional. Time evolution is given by $4$-cobordisms and the space-time is naturally thought of as a $4$-dimensional manifold, as it is intuitively done.\\

In this picture, the exchange of bosons happens on $1$-handles that are inserted to merge braids. Suppose, for instance, that the particles represented by braids $b_1$ and $b_2$ undergo a merging induced by an admissible cobordism. Then we can slide twists of ribbons from one braid to the other, corresponding to boson exchange of type $W^{\pm}$ and $Z^0$. A photon is exchanged if there is no twisting exchange between the two braids, since the braid representing a boson $\gamma$ is trivial. We allow exchange of braiding between two particles on the handle connecting them. Dynamics is, therefore, thought of as a process in which we perform surgery on the base (space-like) $3$-manifold, exchange braiding and twisting between braids corresponding to particles, and perform opposite surgery again. In general, depending on what twisting and braiding exchange is performed during the surgery, we might produce different particles than those we started with. We assume that, during the dynamics, the restrictions on twisting and braiding considered in the Bilson-Thompson's model are still respected.\\

Our view of braided matter and its interaction through surgery of manifolds intrinsically provides a reason for the stability of fundamental particles. Decay happens on handles, so that a particle (braid) that is embedded in the $3$-manifold $M$ does not intersect any handles. Decay for these particle is automatically ruled out, so that they would not split into their constituent preons. \\

Generally speaking, representing dynamics as cobordisms in $4$-dimensions implies that the diagrams representing a braided network $\mathcal N$ need to be adapted to take into account the extra dimension. First, we assume a ``sliceness'' condition on braided networks $\Gamma\subset \mathcal N$. For trivial cobordisms $\mathcal M_{\mathcal L}^3\times I$ this means that for each $t\in I$, the surface $\Gamma_t\hookrightarrow \mathcal M_{\mathcal L}^3\times t$ at the instant $t$ is embedded and it is isotopic to the initial $\Gamma = \Gamma_0$. So, along trivial cobordisms we have that $\Gamma$ varies isotopically, which is represented diagrammatically by Matsuzaki's moves. For non-trivial cobordisms given above, we assume that when attaching handles the braided surface is isotopic to the initial $\Gamma$ outside the handle where braid exchange is happening. This, in particular, means that the surface is locally not isotopic to $\Gamma$. Dynamics is not encoded via moves that preserve the underlying topology. \\

The mid-manifold $W$ corresponds to a boson exchange, and we can think of it as being virtual, as the handle it lives on is not embedded in the base $3$-manifold, as the particles given by braided matter. Using the notion of Kirby diagrams it is possible to diagrammatically describe $1$-handle surgery and braid interaction. \\

When no braid interaction happens, i.e. when no handle is attached to the base $3$-manifold $\mathcal M$, then the isotopy class of the braided network, along with a framed link $\mathcal L$ up to Kirby moves determines the system completely. In fact, time evolution is described by a trivial cobordism where each slice $\mathcal M\times t$ is homeomorphic to the base manifold $\mathcal M$, which is determined by the framed link $\mathcal L$ up to Kirby moves. \\

In order to describe diagrammatically braid interaction, the surgery needs to be described in four dimensions, since during the handle attachment the slice of space time is not homeomorphic to the base manifold anymore. The $4$-dimensional cobordism can be described by means of a {\it Kirby diagram} \cite{4Manifolds}. In a Kirby diagram, a $1$-handle is depicted by drawing two disjoint balls that represent the attaching spheres of the handle. Then a braided network can be drawn unperturbed far from the $1$-handle, while the two braids that merge are not drawn, as they interact on the $1$-handle.

\section{Isotopy invariants of braided networks}

The fact that isotopy classes of trivalent braided networks are used in the scheme described above implicitly raises two important issues. The first relates to the possibility of associating conserved quantities to particles embedded in braided networks. Secondly, we ideally want to consider invariants that distinguish different particles in Bilson-Thompson's model, and different configurations of particle networks. In the previous literature the invariant link has been used to argue that particles correspond to different states \cite{BT2}. This is a link-valued invariant of braided networks, in the sense that it associates a topological link to a given braided network. In \cite{JH}, the author has constructed certain topologies on braided networks that distinguish different topological configurations associated to particles in Bilson-Thompson's model. We follow a different paradigm in the present article, which is particularly suitable to distinguish different braided network configurations. We utilize, in fact, algebraic invariants of isotopy classes of braided networks based on Kauffman-Lins recoupling theory \cite{KL}. We provide a formulation of invariants closely related to those of \cite{Yokota,Handlebodies}, where we do not preserve Reidemeister move I, but rather we have invariance under the framed Reidemeister move I. In fact, the invariants given by Mizusawa and Murakami in \cite{Handlebodies} are suitable as well in the precise formulation thereby given, although twists would be neglected and the particles in Bilson-Thompson model would be indistinguishable. Topologically, the isotopy classes of Bilson-Thompson particles can be distinguished by their combed position. Since during the evolution with respect to time the bradied network might be substantially deformed (for instance by $1$-handle slides), it might not be straightforward to individuate the particles in a braided network, and having some global factor in the invariant provides information that might not be immediately retrived.  One further subtlety concerns the fact that Matsuzaki's moves treat embeddings in the $3$-space $\mathbb R^3$, or its compactification $\mathbb S^3$. For different manifolds, which are not simply connected, some differences arise, such as the fact that there are different types of unknotted surfaces (depending on which holes they wrap around), but the local moves remain unchanged, and the same paradigm can be applied.\\

Now, we show that classical results of Kauffman and Lins, suit well our machinery. We fix a natural number $n\geq 1$, and let $q$ be a $2n^{th}$ root of unity. We let each trinion of a trivalent braided network correspond to a spin network vertex where the Jones-Wenzl projector is used \cite{KL}, Section 4, Definition 3 with the integer triple $\{a, b, c\}$ which is $q$-admissible, i.e. $a+b+c$ is even, $a+b-c, b+c-a, c+a-b\geq 0$ and $a+b+c \leq 2n -4$. In this perspective, each ribbon of a ribbon network consists of $a$ edges for some $a$, and the crossings are ``cabled'' as well. We consider all $q$-admissible triples at each vertex, and normalize the vertices by dividing by their $\theta$ values divided by the traces $\Delta_i$ of the projectors at the incident edges. In fact, the trinions so considered are orthonormal bases, as seen in \cite{Lick2}, of the skein vector space with three boundary arcs. See also \cite{Yokota}, and compare with \cite{Handlebodies}, where this construction appears as well. 
 Then, using the skein relation for the Jones polynomial at a root of unity $A$ (related to $q$ via $A^2 = q$ as usual), and the recursive definition of the Jones-Wenzl projector, we can associate a numerical value $KL_n (\mathcal N)$ to any braided network as follows. Each twisted edge is written as a power of $A$ times a double half twist, as observed before, and each vertex is multiplied by its conjugate, determined as in \cite{Lick2,Yokota} by taking the mirror image and multiplying the results of smoothing all crossings and iterating the Jones-Wenzl projectors. For each fixed choice of admissible triples, this value has been proven to be a framed isotopy invariant of trivalent graphs up to powers of the value $A$. See \cite{KL}, Section 4. We argue that indeed this is an isotopy invariant of braided networks, by verifying that Matsuzaki's moves are preserved. Moreover, the extra powers of $A$ (or $q$) that appear in the statement of the original result of Kauffman and Lins, in our setting, play the role of counting the twists of the edges. In what follows, we will use a notational convention to refer to braids in our computations. Namely, we will suppose the braids being placed vertically as in Figure~\ref{helonmodel} and Figure~\ref{bosons}, and algebraically described as composition of twists  and crossings as in a representation of the framed braid group. A twist will be denoted by the symbol $\theta$, while a crossing is denoted by $R$. Horizontal juxtaposition is indicated as a tensor product, and the straight ribbon (i.e. the identity element of the framed braid group) is indicated by $|^{2n}$, where $2$ refers to the two edges of the ribbon. For example, with $m=2$, the positron on the top-left of Figure~\ref{helonmodel} is given by $(|^2\otimes R)(R^{-1}\otimes |^2)\theta^{\otimes 3}$, where the composition symbol has been omitted, and reading from right to left corresponds to braids from top to bottom. \\

\begin{theorem}\label{thm:invariant}
Let $\mathcal N$ be a trivalent (oriented) braided network and let $n\geq 1$. Then, $KL_n (\mathcal N)$ is an isotopy invariant of $\mathcal N$.
\end{theorem}
\begin{proof}
We need to verify that $KL_n (\mathcal N)$ preserves Matsuzaki's moves since, as observed above, braided networks are isotopic if and only if their diagrams are related by a finite sequence of these moves. Invariance under framed Reidemeister moves II and III are the fundamental results of Kauffman in the construction of the Jones polynomial from the (now called) Kauffman bracket polynomial. Sliding ribbons above and below trivalent networks is seen to preserve the value of $KL_n$ as a consequence of framed Reidemeister moves II and III, once the recursive definition of Jones-Wenzl projector is applied at the vertex. In fact, after smoothing, one obtains that the lines of a ribbon overpass/underpass closed Jordan curves, or cups and caps. Either way, using the Reidemeister moves II and III one can slide the ribbon as required. This is also observed in Section~4.3 of \cite{KL}, and it works for any fixed admissible colouring. We need to verify the framed Reidemeister move I (i.e. cancellation of kinks). This does not appear, at least in the form needed for our purposes, explicitly in the work of \cite{KL}. We proceed as follows. First, consider a single strand self-crossing to give a positive twist. Smoothing at the single crossing we obtain $A\cdot \circ | + A^{-1}\cdot |$, where $\circ$ indicates a closed Jordan curve, $|$ denotes a straight line and we use a $\cdot$ symbol to separate scalars from diagrams they refer to. Using the renormalization value for $\circ = -A^2-A^{-2}$, we find the value $-A^3\cdot |$. Applying the same procedure to the self-crossing that generates negative loops, one obtains $A^{-3}\cdot |$. So, concatenation of positive and negative kinks gives the straight strand $|$, as expected. Now, in general, the ribbons of the braided network $\mathcal N$ correspond to groups of more strands, so the positive loops consist of parallel strands that twist together to give a self-crossing. We show how to proceed with $m=2$ strands, since a simple induction generalizes this step to all values of $m$. Using the computation for the positive and negative twists just obtained, we see that a single line can be slid above and below a kink, since both evaluations simply give $-A^{\pm3}\cdot R$, where $R^{\pm}$ indicates a crossing of two lines (positive or negative). A self-crossing with two strands twisted positively, which we indicate with the symbol $\theta_2$, can therefore be rewritten as the tangle
\begin{equation}
   \theta_2 =  (d\otimes |)(|\otimes |\otimes R)(|\otimes R\otimes |)(b\otimes |\otimes |)(\theta_1 \otimes |),
\end{equation}
where $b$ is the birth (i.e. cap) tangle, $d$ is the death (i.e. cup) tangle,  $\theta_1 = (d\otimes |)(|\otimes R)(b\otimes |)$ is the decomposition of (positive) twist tangle into compositions of generators of the tangle category. Using $\theta = -A^3\cdot |$, and by Reidemeister moves II and III, the previous expression becomes $-A^6\cdot R^2$. Similar computations show that the negative twist is written as $-A^{-6}\cdot R^{-2}$. The concatenation of positive and negative twists gives the straight lines $||$. In the general case for $\theta_{m}$, a self-crossing of $m$ strands, one obtains
\begin{equation}
  \theta_{m} = (-1)^{3m}A^{\pm 3m} \cdot (|^{m-2}\otimes R^\pm)\cdots(R^\pm\otimes |^{m-2})(R^\pm\otimes |^{m-2})\cdots(|^{m-2}\otimes R^\pm),  
\end{equation}
where $+$ or $-$ depends on whether we consider positive or negative twists, respectively, where the symbol $|^m$ indicates $m$ straight parallel lines. It follows that the composition of twists for arbitrary $m$ gives $|^m$. In this case, as before, the argument works for any fixed admissible colouring. To show invariance under IH move, one considers that this corresponds to a change of an orthonormal basis to another one, and therefore consists of applying a unitary matrix, see \cite{Lick2}, and the value of $KL_{n}(\mathcal N)$ does not change. In this case we need to let the colourings vary over all the admissible possibilities. 
\end{proof}

\bigskip 

\subsection{Examples and Computations}

The computation of twists in Theorem~\ref{thm:invariant} suggests that the invariant $KL_n$ distinguishes, in general, the particles of Bilson-Thompson model, in the sense that two braided networks that are identical, except for two regions that enclose one of the particles of Figure~\ref{helonmodel} or Figure~\ref{bosons}, generally have different values of $KL_n$. Of course, in order to give a specific statement of this sort, regarding what particles are distinguished, one needs to fix the values of $q$ (i.e. $A$) and $n$. It is not the scope of the present article to fix values for them, but we rather are contented with showing that the theory is nontrivial. \\

We observe that no isotopy invariant of braided networks is suitable to distinguish $Z^0$ and $\gamma$ bosons, as they differ by an application of the framed Reidemeister move I which, by construction, is preserved by any isotopy invariant of braided networks. \\

By means of example we show that the invariant $KL_4$ does change when we replace a boson $\gamma$ by a boson $W^+$, so that the algebraic setting detects the topological non-equivalence of different bosons. Setting $n=4$, means that each ribbon edge is represented by parallel lines whose total sum does not exceed $6$. Since the two networks are assumed to be identical outside the region enclosing the two particles $W^+$ and $\gamma$, we can substantially compute the invariant $KL_4$ of $W^+$ and $\gamma$ as depicted in Figure~\ref{bosons}, as the computation of $KL_4$ would give identical values, when computing the rest of the braided network. We rewrite the twists of $W^+$ using the computation performed in the proof of Theorem~\ref{thm:invariant}. The computations are identical for each colouring, with the only difference being that one needs to take into account the number of strands in the twists. For example, trinions all coloured with two strands (the triple $\{2,2,2\}$ is $q$-admissible), gives us an overall factor of $A^{18}$ and the braids of $W^+$ are replaced by $R^2\otimes R^2\otimes R^2$, two consecutive crossings that braid the edges of each ribbon. We perform now smoothing of each crossing and utilize the definition of the Kauffman bracket polynomial. Observe that performing the two smoothings on a ribbon produces two diagrams, one with a straight ribbon and an extra factor of $A^2$, and one with the ribbon removed and a factor of $1-A^{-4}$. The computation is therefore reduced simply to the product of these factors. We denote by $|^i$ the braid with $i$ straight ribbons, where $i = 0,1,2,3$. We obtain that the braid $W^+$ can be replaced in the braided network by the sum of terms 
\begin{equation}\label{eqn:bosonW}
   W^+ = -A^{24}\cdot |^6 - A^{18}\cdot [A^2(1-A^{-4}) + 2A^4(1-A^{-4})]\cdot |^4 - A^{18}[(1-A^{-4}) + 2A^2(1-A^{-4})^2]\cdot |^2 - A^{18}(1-A^{-4})^3\cdot  |^0.
\end{equation}

By definition, we have that the evaluation with the boson $\gamma$, with the same colour, is simply $|^6$, which appears up to a factor in Equation~(\ref{eqn:bosonW}). When computing $KL_4$ of $W^+$, the terms with lower powers of $|$ disappear, since they contain turnbacks that are annihilated by the Jones-Wenzl projectors. In general, considering all colourings, $KL_4$ distinguishes the two bosons, as previously claimed. We remark that in the previous computation we have pairs of ribbons or single ribbons, instead of triples representing framed braids. In general, this is a middle step of a computation and, therefore, it is not to be considered as a physically meaningful configuration of a braided network.\\

It has been pointed out, in \cite{BT2}, that in order to introduce a nontrivial dynamics for braided particles, it is necessary to consider moves such that (what we have here called) braided networks do not constitute noiseless subsystems. In practice, what the authors of the article \cite{BT2} have realized, is that it is not enough to consider braided networks as arising from triangulations of manifolds and only local moves that are derived from Pachner moves. A natural question to ask, is whether performing surgery on manifolds to introduce dynamics represents an example of such a paradigm. We argue that this is indeed the case, and that the invariant $KL_n$ shows this. In fact, (very tedious) computations show that there are cases where the interaction/decay of two particles in a braided network, obtained via surgery as described above, produces another braided network whose $KL_n$ invariant is different from the initial one. Applying Theorem~\ref{thm:invariant} it follows that initial and final state are not isotopy equivalent and, therefore, the topology of the system has been changed during surgery, giving an affirmative answer to the previous question.  \\

We briefly describe the computations related to electron/positron annihilation producing one virtual photon. Recall that, in the setting of the present article, virtual means that the photon produced in the process lives on the handle that mediates the interaction. We set $n = 4$, and consider the colouring where all edges are assigned value $2$ (the remaining colours are computed in the same way). Let us indicate the value associated to this colouring by $KL^2_4(\mathcal N)$. Suppose that $e^+$ and $e^-$ are right-handed. First, observe that the value of two braided networks that are identical outside two regions containing $e^+$ and $e^-$ differ by a multiplication of a (Laurent) polynomial $f(A^{\pm 1})$, where $A$ refers to the positron, and $A^{-1}$ refers to the electron. Therefore, to compute $KL^2_4(\mathcal N)$, for a braided network $\mathcal N$ containing a positron and an electron, we can perform the computation of $KL^2_4(e^-)$ for the electron, and obtain the contribution of $e^+$ simply by changing the variable $A^{-1}$ to its inverse $A$. This computation can be simplified as follows. One first uses the formula obtained in Theorem~\ref{thm:invariant} to rewrite the twisting of each ribbon as composition of two negative crossings. This produces an overall factor of $-A^{-18}$ multiplying the braid corresponding to $e^-$. Now, we obtain a braid in pure twist form by applying two rotations of trinions \cite{BT2009}. We remark, here, that we can perform this operation by means only of framed Reidemeister moves, without twisting the vertex, or using the IH move, since the number of twists is even (due to the orientability assumption discussed above) \cite{Matsuzaki}. No overall multiplying factor is introduced during this process. We can also see this by applying Proposition~3 in Section 4.2 of \cite{KL}, since two opposite rotations would be utilized. In pure twist form, the right handed electron $e^-$ can be written as $[0,-2,-1]$, see table at page 12 of \cite{BT2009}. Now, we use the computation for the composition of two negative crossings $R^{-2} = A^{-2} \cdot |^2 + (1-A^4)\cdot |^0$, as obtained in Theorem~\ref{thm:invariant}, to complete the computation. This requires three iterations to give the result 
\begin{equation}\label{eqn:e-}
 KL^2_4(e^-) = -A^{-24}\cdot KL^2_4( |^6) + A^{-18}Q(A)\cdot KL^2_4(|^4) + A^{-18}P(A)(1-A^4)\cdot KL^2_4(|^2),   
\end{equation}
 where 
 \begin{equation}\label{def:PQ}
     P(A) := A^{-18}[2A^{-2}(1-A^4) - (1-A^4)(-A^2-A^{-2})],\ \  Q(A) := A^{-22}(1-A^4) - P(A)A^{-2}.
 \end{equation}
 
  As previously observed, we also have 
 \begin{equation}\label{eqn:e+}
  KL^2_4(e^+) = -A^{24}\cdot KL^2_4(|^6) + A^{18}Q(A^{-1})\cdot KL^2_4(|^4) + A^{18}P(A^{-1})(1-A^{-4})\cdot KL^2_4(|^2),  
 \end{equation}
 where the polynomials $P$ and $Q$ are defined in Equation~(\ref{def:PQ}).
  Let $\mathcal N$ be a braided network containing $e^-$ and $e^+$ and suppose that, via surgery on the base manifold, $\mathcal N$ undergoes an annihilation of $e^-$ and $e^+$ on an attached handle (over which a virtual boson is exchanged). Let us denote by $\mathcal N'$ the braided network obtained after surgery, and by $\hat{\mathcal N}$ the network obtianed from $\mathcal N$ by removing $e^+$ and $e^-$ and replacing them with two particles $\gamma$. From the computations giving Equation~(\ref{eqn:e-}) and Equation~(\ref{eqn:e+}) we have that 
\begin{equation}\label{eqn:annihilation}
 KL^2_4(\mathcal N) = KL^2_4(\hat{\mathcal N}) + KL^2_4(\text{ terms\ containing}\ \ Q(A^\pm), P(A^\pm))\,.   
\end{equation}
 
 The extra terms in Equation~(\ref{eqn:annihilation}) correspond to the braids with two ribbons, or a single ribbon, after the smoothings of $e^\pm$, and their coefficients do not generally vanish identically, but their contribution to $KL_4$ evaluations are trivial, since they contain turnbacks annihilated by the Jones-Wenzl projector. Summing the contributions due to the other $q$-admissible colourings, we see directly that the surgery corresponding to annihilation electron/positron changes the isotopy class of the braided network since, in general, we have that $KL_4(\hat{\mathcal N})$ is different from $KL_4(\mathcal N')$. 
 
\section{Possible pre-geometric unification of matter and gravity}
\label{Sec:Pre-geometry}
In this final section, we speculate how the emergence of matter may be naturally related to the emergence of gravity, out from the Planck realm. Here the quantum fluctuations of the fields destroy the common notion of space-time arena, and call for a totally new theoretical perspective. While considering matter to be constituted by braids of pre-geometric fields, one assumes automatically a framework that is the same as that adopted in emergent gravity scenarios. Thus gravity and matter can be finally treated on the same footing. At the same time, describing matter requires the possibility to account for particles interaction. Letting features of matter remain preserved by the (pre-geometric) dynamics, prevents from instantiating a physical description of particle interaction. A nontrivial dynamics can be rather introduced, considering moves for which braids, now representing particles, do not constitute noiseless subsystems, and the underlying topology is therefore not preserved. \\

Several concepts percolate into the definition of particles. In general, particles are defined asymptotically, on Cauchy-surfaces, within the boundary formalism. Resorting to cobordism, particles, as any states, must be represented on the boundary of the 2-complexes associated to the space-time bulk. Within this picture, braids representing real particles are localised on the boundaries. Selection rules apply to the big picture, requiring a translation in terms of topological invariants. The preservation of the associated quantum numbers then instantiates the rules governing the dynamical processes that are allowed within the pre-geometric frame-work. The example provided by the electron-positron annihilation elucidates this concept concerning the preservation of the quantum lepton number. \\  

Matter and gravity arising from pre-geometry calls for a profound rethinking and a deeper understanding of the emergence of structures and forces, since particles are now thought of in terms of their complexity, no more depicted without extension, and gravity arises due to interacting matter, its source and reason of being. A phase transition can be naturally thought to be advocated in order to accomplish for the emergence of braiding in the pre-geometry. A way this phase transition can be realized is by accounting for a thermal flux that, while bringing the pre-geometric entities to equilibrium, restore at lower energy scales the quantum constraints that define gravity, and through the scalar constraint automatically instantiate braided matter interaction. \\

On the one side, it can be argued that the handle prescription we have introduced in this paper is a non-local instantiation of the scalar Hamiltonian constraint, when irreducible representations of the holonomies assigned to the links of the 1-complexes are quantum group like. Then ribbons and fat graphs are taken into account, and the scalar Hamiltonian constraint may generate the handle attachment. On the other side, it is possible to think of the very same emergence of braided matter in terms of a quantum dynamics that breaks diffeo-invariance, out of the equilibrium, while the dynamical symmetry breaking is taking place. Technically, one might account in different ways for the evolution of holonomies and Wilson lines, so to encode braids nucleation, including dual techniques borrowed from either string theories or stochastic quantum field theory methods. \\      

The emergence of gravity seems therefore to be attained differently than for braids. The latter ones are rather inserted through the action of knotting and braiding  operators, which still preserve triangulation invariance. Gravity, on the other side, arises as the byproduct of the interaction of particles. In other words, it only materialises when we consider particle interaction, which can be only made possible by handle attachment. On the other hand, handle attachment is instantiated at the quantum level by the action of the Hamiltonian constraint, so it traces back to the consideration of the full constrained system that realize the symmetries and the dynamics of gravity. The seems to be on the same line than the considerations made by Bilson-Thompson, Markopoulou and Smolin in \cite{BT2}, at the end of the article. Matter remains a gravitational noiseless subsystem, as long as particle interaction, i.e. handle attachment, is considered. \\ 

According to this picture, the emergence of matter is then ascribed to the formation of braids, into portions of the braided network. As reminded by Rovelli \cite{Rovelli:2010wq}, spin networks are not usually considered to be embedded in $3$-manifolds, but rather intended as abstract graphs, on which knotting and braiding act without need of any physical notion. Within the Bilson-Thompson model, the only available physical explanation is actually the one that provides the base to the concept of particles, and hence allows for the emergence of matter. The very same formation of braided matter, which requires the use of quantum groups, also naturally instantiates an infrared regulator for quantum gravity. To unveil this phenomenon, it is essential to remind that the deformation parameter, when is a root-of-unity, automatically limits to a maximum weight the irreducible representations of quantum groups involved, thus preventing bubble divergences. \\ 

The deformation parameter of the quantum groups adopted in $3$-dimensional and $4$-dimensional quantum gravity has been hitherto conjectured to be connected to the cosmological constant. But the presence of a bare cosmological constant is inessential to the physics, since several quantum mechanisms can be advocated to get rid of its classical value. 
At the same time, possible explanation involve also BCS condensate of fermionic matter, and a wealth of other alternatives that encode bosonic condensates and cosmological plasma physics. Nonetheless, the perspective we are pushing forward here is that the regulator is rather justified by the emergence of matter, described by the quantum groups. In other words, matter is the natural regulator of gravity, and this becomes evident due to the quantum groups representations that are required by the braided matter. \\ 

\section{Discussion}
\label{Sec:Discussion}

How matter is to be incorporated into a viable theory of quantum gravity remains an important open question. One interesting proposal, pursued by several authors, is to embed the topological toy model of Bilson-Thompson in which leptons and quarks are represented as simple braids, within a class of background independent theories of quantum gravity based on spin networks (or their generalisation). One major obstacle in the development of this novel framework has been the lack of any suitable dynamics reminiscent of the particle interactions observed in the laboratory. It turns out that the Pachner moves which generate the dynamics on spin networks do not give rise to the dynamics governing particle interactions. The main theoretical result of this paper has been to demonstrate, by supplementing the local Pachner moves with a non-local graph changing 1-handle attachments, that it is possible to introduce suitable particle interactions topologically via the braid product, whilst preserving the original simple classification of fermions proposed by Bilson-Thompson.\\

While our result develops an appealing theoretical framework, we stress that at this early stage we do not have a satisfactory or complete dynamical theory. It is not clear at this point in time how specific amplitudes for particle processes may be calculated within this purely topological construction, nor how the Higgs boson, or PMNS and CKM matrices may arise. Ultimately, any convincing theory of matter and quantum gravity must address these important questions, and make actual predictions that can be tested experimentally. Despite these current shortcomings, we believe this approach to pre-geometrically unifying matter with gravity is worth investigating further, and only a more complete study of this framework can establish its viability.\\

Among the points to bear in mind for forthcoming investigations, the issue of implementing a covariant representation that encodes quantum fluctuations, without requesting the splitting among space and time, is one of the most urgent. A possible way to overcome this problem could be considering a braided surface that is embedded in a $4$-dimensional space. Deepening the link among braided matter and the emergence of an infrared regulator for gravity is as crucial as the previous point to this line of research, and will deserve detailed forthcoming studies. Finally, demonstrating that the handle attachment hereby introduced instantiates a non-local realization of the scalar Hamiltonian  constraint, would be a major corroboration of the idea that the emergence of matter and gravity are intrinsically related.

\section*{Acknowledgement} 
 We thank the anonymous Referee for comments and suggestions that helped improve the manuscript. 
NG was supported by the Natural Science Foundation of the Jiangsu Higher Education Institutions of China Programme Grant 19KJB140018 and XJTLU REF-18-02-03 Grant. AM wishes to acknowledge support by the NSFC, through the grant No. 11875113, the Shanghai Municipality, through the grant No. KBH1512299, and by Fudan University, through the grant No. JJH1512105. EZ was supported by the Estonian Research Council through the grant MOBJD679.

\end{document}